\documentclass[11pt,twoside]{article}
\usepackage[round]{natbib}
\usepackage{float}
\usepackage{geometry}
\usepackage{array}
\usepackage{makecell}
\input{commands-mm}

\begin{document}

\begin{center}

  {\LARGE
  \textbf{Signal recovery in the polychromatic computed tomography model: Injectivity and complexity}
  
  }

\vspace*{.2in}

{\large{
\begin{tabular}{ccc}
Xuanzhou Chen$^\dagger$, Ashwin Pananjady$^{\star,\dagger}$
\end{tabular}
}}
\vspace*{.2in}

\begin{tabular}{c}
Georgia Institute of Technology, School of Industrial and Systems Engineering$^\star$ and \\
School of Electrical and Computer Engineering$^\dagger$ \end{tabular}

\vspace*{.2in}

\today

\vspace*{.2in}

\begin{abstract}
We consider a nonlinear model motivated by polychromatic computed tomography (CT). Here, $W$ distinct $d$-dimensional signals $x^*_1, \ldots, x^*_W \in \mathbb{R}^d$ must be recovered from $n$ measurements $(a_i, y_i)_{i = 1}^n$ that obey the nonlinear model $\mathbb{E}[y_i|a_i] = h(\inprod{a_i}{x^*_1}, \ldots, \inprod{a_i}{x^*_W})$, where $h: \mathbb{R}^W \to \mathbb{R}$ is a known nonlinearity that models a certain type of exponential attenuation law. Even when there is no noise in the measurements, the sample size $n$ (for any measurement ensemble $\{a_i\}_{i = 1}^n$) must exceed the number of unknowns $Wd$ to guarantee that the underlying signals are identifiable. We construct a measurement ensemble that ensures perfect signal recovery almost surely provided $n \geq Wd + W - 1$, thereby isolating the injectivity threshold up to an additive factor $W - 1$. We also study computational complexity of signal recovery in this model with a general measurement ensemble $\{a_i\}$. We show that if $W \geq 2$, then there is a measurement ensemble for which deciding whether there exist signals consistent with the measurements is NP-hard. In particular, this implies that polychromatic, multimaterial CT reconstruction is NP-hard in general. This finding stands in sharp contrast to the single-material setting, for which a polynomial-time algorithm can provably perform signal recovery for any measurement ensemble.
\end{abstract}
\end{center}

\section{Introduction}
Consider the polychromatic computed tomography (CT) problem~\citep{ 
tang2023spectral, https://doi.org/10.1002/mp.15621},
where the forward model is given by
\begin{align} \label{eq:model}
\mathbb{E}[y_i|a_i] = \sum_{j = 1}^W s_j \exp(- \inprod{a_i}{x^*_j}) \quad \text{ for } i = 1, \ldots, n
\end{align}
and for each $j \in [W]$, the vector $x^*_j \in \real^d$ denotes the unknown absorptivities across the $d$ voxels for the $j$-th wavelength. These vectors are distinct and the spectral sensitivity coefficients $\{s_j\}_{j = 1}^W$ are all strictly positive\footnote{If not, we obtain a model with fewer than $W$ wavelengths.}. The inverse problem is to recover the unknown vectors $\{x^*_j\}_{j = 1}^W$ from measurements $(a_i, y_i)_{i = 1}^n$, assuming we have perfect access to $s_1, \ldots, s_W$.  Without loss of generality, we will additionally assume that $\sum_{j = 1}^W s_j = 1$, so that $(s_1, \ldots, s_W)$ can be viewed as a probability mass function supported on $W$ points. 

In real photon-counting CT, the scalar $y_i$ is the number of photons that are incident at the detector and this is modeled as a random variable obeying the conditional expectation relation~\eqref{eq:model}. In addition, the vectors $a_i$ are obtained from discretizing line integrals in the Radon transform.
However, given that our focus in this work is on understanding the fundamental limits of signal recovery from the mathematical forward model~\eqref{eq:model}, we study the \emph{noiseless} setting wherein $y_i = \mathbb{E}[y_i|a_i]$ a.s. We will also study the setting in which the measurement vectors $a_i$ are general and unstructured. These assumptions are commonly made in the CT model~\citep{fridovich2026gradient, charisopoulos2025nonlinear} to allow us to isolate information-theoretic and algorithmic issues. 

To develop some intuition from the model, consider the \emph{monochromatic} variant when $W = 1$. In this case, the noiseless model \eqref{eq:model} can be written as 
\begin{align}\label{eq:model_w_is_1}
    y_i = \exp(-\langle a_i, x_1^*\rangle) \quad \text{ for } i = 1, \ldots, n.
\end{align}
Clearly, signal recovery in this model can be performed by logarithmically transforming both sides of Eq.~\eqref{eq:model_w_is_1} and solving a linear system~\citep{tang2023spectral}. Thus, the unknown signal $x_1^*$ can be recovered as long as $n \geq d$ and the matrix $A$ formed by stacking the vectors $\{a_i\}_{i = 1}^n$ as rows has full column rank.\footnote{In the more realistic version of model~\eqref{eq:model_w_is_1} with Poisson noise, signal recovery can be performed via maximum likelihood estimation (i.e. Poisson regression), which corresponds to a convex optimization problem.} Taken together, these observations resolve two issues in the monochromatic model when $W=1$. First, the injectivity threshold is $n = d$, in that there exist $d$ measurements $(a_i, y_i)_{i = 1}^d$ from model~\eqref{eq:model} for which exact signal recovery can be performed. Second, the computational complexity of signal recovery is polynomial-time, in that signal recovery can always be performed in polynomial time for any measurement ensemble. 

In addition to the chromaticity of the forward model, there is another consideration that aids signal recovery in CT applications:\textit{ multimaterial decomposition}. This decomposition relies on the fact that the unknown vectors $\{x^*_j\}_{j = 1}^W$ in model~\eqref{eq:model} are all related by a  decomposition into $M$ basis materials. In particular, we posit that for some known and positive absorption coefficients $\{ \mu_{j, m} \}_{j \in [W], m \in [M]}$, we have $x^*_j = \sum_{m = 1}^M \mu_{j, m} \overline{x}_{m}$, where $\overline{x}_{m}$ collects the density of material $m$ across the $d$ voxels. If the number of materials is less than the number of wavelengths (i.e. $M < W$), then we have a compressed representation of the unknowns. If $M = W$ and the matrix formed by the absorption coefficients $\{ \mu_{j, m} \}_{j \in [W], m \in [M]}$ is invertible, then recovery of $\{\overline{x}_{m}\}_{m = 1}^M$ is equivalent to recovering $\{x^*_j\}_{j = 1}^W$. On the other hand, if $M > W$, then  $\{\overline{x}_{m}\}_{m = 1}^M$ are not identifiable even if $\{x^*_j\}_{j = 1}^W$ are, since the absorption coefficient matrix has a nontrivial nullspace.

A special case of the above decomposition is the \emph{single-material} setting, in which $M = 1$ (and $W$ is any natural number).
Recent work~\citep{lou2026accurate,kim2026perfusion} showed that in this setting, the problem of recovering the signal $\overline{x}_1$ is exactly solvable in polynomial time via a variational inequality formulation. The only requirements for this result are that the coefficients $\{s_j, \mu_j \}_{j = 1}^W$ are positive, and that the map $\overline{x}_1 \mapsto A \overline{x}_1$ is injective. This result complements the monochromatic case~\eqref{eq:model_w_is_1}. As mentioned above, when $W = 1$ and $M \geq 2$, the multimaterial recovery problem is ill-posed and we never have injectivity. Nevertheless, the problem of finding a feasible solution is polynomial-time solvable because the Poisson negative log-likelihood is convex in this case~\citep{barber2024convergence}.

\begin{table}[H]
\centering
\caption{Signal recovery in the model~\eqref{eq:model} with the
low-rank material decomposition $x^*_j = \sum_{m=1}^M \mu_{j,m}\bar{x}_m$, where $d$ is the number of voxels and $n$ the number of measurements. Each cell reports the injectivity threshold (top) and the computational status (bottom) known from prior work.}
\small
\setlength{\tabcolsep}{4pt}
\renewcommand{\arraystretch}{1.25}
\setlength{\extrarowheight}{1pt}
\begin{tabular}{|l|c|c|}
\hline
& \makecell{Single material\\ $M = 1$}
& \makecell{multimaterial\\ $M \geq 2$} \\
\hline
\makecell[l]{Monochromatic\\ $W = 1$}
& \makecell{$n \geq d$ (tight)\\ poly.\ time}
& \makecell{Never injective\\ poly.\ time} \\
\hline
\makecell[l]{Polychromatic\\ $W \geq 2$}
& \makecell{$n \geq d$ (tight)\\ poly.\ time}
& \makecell{\textbf{open}\\ \textbf{open}} \\
\hline
\end{tabular}
\label{tab:landscape}
\end{table}

Given this state of affairs (summarized in Table~\ref{tab:landscape}), our focus is on studying injectivity and complexity of signal recovery in the polychromatic CT model~\eqref{eq:model}.
Concretely, we study the following questions: \\
(Q1) What is the injectivity threshold of the polychromatic CT model? \\
(Q2) Is signal recovery in polychromatic CT solvable by a polynomial-time algorithm? \\
We answer question (Q1) by designing a random measurement ensemble $(a_i)_{i = 1}^n$ 
with $n = Wd + W - 1$ such that the $W$ unknown vectors $\{x^*_j\}_{j = 1}^W$ can be recovered from measurements $(a_i, y_i)_{i = 1}^n$ with probability $1$
(see \Cref{subsec:injectivity}). Since 
$n \geq Wd$ is a necessary condition given the number of unknowns, this isolates the injectivity threshold up to an additive factor of $W - 1$. We then address question (Q2) with a negative result, showing that signal recovery in the polychromatic CT model~\eqref{eq:model} is not solvable in polynomial time unless $\mathsf{P= NP}$. Our result is shown when $W = 2$: we establish via reduction from the $\mathsf{NP}$-hard \textsc{Partition} problem that deciding whether there exists a solution $\{x^*_j\}_{j = 1}^W$ consistent with a set of measurements $(a_i, y_i)_{i = 1}^n$ is already $\mathsf{NP}$-hard (see \Cref{subsec:computation}). Our main results are stated in \Cref{sec:main-res} and \Cref{sec:proofs} contains all proofs.
Taken together, our results resolve the open regimes alluded to in Table~\ref{tab:landscape}.
We conclude with a discussion of open problems in \Cref{sec:conclusion}.

\section{Main results}\label{sec:main-res}
    As mentioned above, we study both injectivity and complexity of the polychromatic CT model~\eqref{eq:model}.
Our results on injectivity are presented first.
\subsection{Injectivity threshold of  polychromatic CT}\label{subsec:injectivity}

Since there are $Wd$ unknown real values in the model~\eqref{eq:model}, a necessary condition for signal recovery is $n \geq Wd$. We now show that there exists a  measurement ensemble guaranteeing unique recovery provided we have the slightly stronger condition $n \geq Wd + W - 1$. Our proof is constructive: 

Choose a uniformly random $d \times d$ unitary matrix $R$ and let $r_1, \ldots, r_d$ denote its columns. Construct the measurement vectors as
\begin{subequations} \label{eq:measurement-ensemble}
\begin{align} 
a_{k} &= k r_1 \qquad \;\;\text{for } k = 1,\dots,2W-1, \label{eq:stage1}\\
a_{\ell W + k} &= r_\ell + k r_1 \;\; \text{for } \ell = 2,\dots,d; k = 0,\dots,W-1. \label{eq:stage2}
\end{align}
\end{subequations}
The total number of measurements in the ensemble~\eqref{eq:measurement-ensemble} is
\[
n = (2W-1) + (d-1)\cdot W = Wd + W - 1.
\]
In the following theorem, we claim that this ensemble can be used to perform signal recovery.
\begin{theorem}\label{thm:injectivity}
Under the random measurement ensemble~\eqref{eq:measurement-ensemble} and observations $(y_i)_{i = 1}^n$ from model~\eqref{eq:model}, the vectors $(x_j^*)_{j = 1}^W$ are uniquely recoverable from $(a_i,y_i)_{i = 1}^n$ with probability 1.
\end{theorem}

Theorem~\ref{thm:injectivity} thus isolates the injectivity threshold of the polychromatic CT model up to an additive factor $W-1$. Note that when $W = 1$, the injectivity threshold is tight. Moreover, it is worth noting that the procedure for recovering the underlying signals from this measurement ensemble is also constructive, and presented in the proof (\Cref{subsec:proof_injectivity}). 

\subsection{Computational complexity of signal recovery}\label{subsec:computation}

Now suppose our task is to recover $\{x^*_j\}^{W}_{j=1}$ in model~\eqref{eq:model} from general measurements $(a_i, y_i)^n_{i=1}$ with a computationally efficient algorithm. In order to accomplish this task, we must be able to efficiently determine, given $(a_i, y_i)^n_{i=1}$, whether there exists a set of vectors $\{x^*_j\}^{W}_{j=1}$ consistent with these measurements. We study this latter question in the special case $W = 2$ with $s_1 = s_2 = 1/2$. In this case, the model~\eqref{eq:model} simplifies to
\begin{align} \label{eq:model-simple}
y_i = \frac{1}{2} \exp( - \inprod{a_i}{x^*_1} ) + \frac{1}{2} \exp( - \inprod{a_i}{x^*_2} ).
\end{align}

\begin{theorem} \label{thm:NP-hard}
Given data $(a_i, y_i)_{i = 1}^n$ with $a_i \in \real^d$ and $y_i \in \real$, the problem of deciding if there exist vectors $x^*_1, x^*_2 \in \real^d$ satisfying the system~\eqref{eq:model-simple} is NP-hard.
\end{theorem}

While Theorem~\ref{thm:NP-hard} is stated for the polychromatic CT model, the result implies $\mathsf{NP}$-hardness of general multimaterial CT. To see this, consider a multimaterial model with $W = 2$ and $M= 2$ and a diagonal absorption coefficient matrix $\mu \in \real^{2 \times 2}$. This results in the model
\begin{align} \label{eq:model-MM}
y_i = \frac{1}{2} \exp( - \mu_{1,1}\inprod{a_i}{\overline{x}_1} ) + \frac{1}{2} \exp( - \mu_{2, 2}\inprod{a_i}{\overline{x}_2} ).
\end{align}
Clearly, if $\mu_{1, 1}$ and $\mu_{2, 2}$ are nonzero, then hardness of recovering the signal $(x^*_1, x^*_2)$ in model~\eqref{eq:model-simple} immediately implies hardness of recovering the signal $(\overline{x}_1, \overline{x}_2)$ in model~\eqref{eq:model-MM}.

\begin{remark}\label{rmk:search}
While Theorem~\ref{thm:NP-hard} is stated for the problem of deciding whether or not a consistent solution exists, finding a pair $(x_1^*,x_2^*)$ consistent with $(A,y)$ is NP-hard even under the promise that a consistent solution exists. This is because the \textsc{Partition} problem remains $\mathsf{NP}$-hard even under a promise. In particular, finding a feasible \textsc{Partition} is $\mathsf{NP}$-hard even under the promise that a feasible solution exists~\citep{papadimitriou1998combinatorial}.
\end{remark}

\section{Proofs}\label{sec:proofs}

We now present the proofs of our two main theorems.

\subsection{Proof of injectivity threshold}\label{subsec:proof_injectivity}

We use two lemmas to prove Theorem~\ref{thm:injectivity}. Lemma~\ref{lem_1} is the key signal recovery lemma using ensemble~\eqref{eq:measurement-ensemble} and relies on an assumption involving the matrix $R$.  Lemma~\ref{lem_2} shows that this assumption holds true with probability $1$ under our randomness condition on $R$. 

\begin{lemma}\label{lem_1}
If the scalars $\langle r_1, x_1^* \rangle, \ldots, \langle r_1, x_W^* \rangle$ are distinct, then the measurement ensemble~\eqref{eq:measurement-ensemble} guarantees unique recovery of the scalars $\langle r_\ell, x^*_j \rangle$ for all $\ell \in [d]$ and $j \in [W]$.
\end{lemma}

\begin{proof}[Proof of Lemma~\ref{lem_1}]
We will perform recovery in two stages.

\begin{enumerate}
\item[(1)] Use the first $2W-1$ measurements~\eqref{eq:stage1} to recover $\langle r_1, x_1^* \rangle, \ldots, \langle r_1, x_W^* \rangle$ by Prony's method.
\item[(2)] For each $\ell = 2, \ldots, d$, solve a linear system using the measurements $\{r_\ell + k r_1\}_{k=0}^{W-1}$ to recover $\langle r_\ell, x_1^* \rangle, \ldots, \langle r_\ell, x_W^* \rangle$.
\end{enumerate}

Define the shorthand
\[
P_{j,\ell} = \exp(-\langle r_\ell, x_j^* \rangle), \qquad j = 1, \ldots, W,\; \ell = 1, \ldots, d.
\]
Note that the map $\langle r_\ell, x_j^* \rangle \mapsto P_{j,\ell}$ is injective, so recovering $P_{j,\ell}$ is equivalent to recovering $\langle r_\ell, x_j^* \rangle$.

\medskip
\noindent\textbf{Stage 1.} From the first stage of measurements $\{k r_1\}_{k=1}^{2W-1}$, we obtain
\[
\sum_{j=1}^{W} s_j P_{j,1}^k \qquad \text{for } k = 1, \ldots, 2W-1.
\]
In addition, since $s_1, \ldots, s_W$ are known, we know $\sum_{j=1}^{W} s_j P_{j,1}^0$.

This is the problem of recovering a measure supported on $W$ points from $2W$ consecutive moments. Prony's method \citep{giesbrecht2006symbolic} will succeed provided $\{P_{j,1}\}_{j=1}^W$ are distinct (which is true by the lemma's assumption).

\medskip
\noindent\textbf{Stage 2.} We can now assume we know $\{P_{j,1}\}_{j=1}^W$ from Stage 1. For each $\ell = 2, \ldots, d$, the second set of measurements gives access to
\begin{align*}
 &\sum_{j=1}^{W} s_j P_{j,\ell}, \quad \sum_{j=1}^{W} s_j P_{j,1} P_{j,\ell}, \\ &\sum_{j=1}^{W} s_j P_{j,1}^2 P_{j,\ell}, \; \ldots, \; \sum_{j=1}^{W} s_j P_{j,1}^{W-1} P_{j,\ell}.   
\end{align*}

This is a linear system in $\{P_{j,\ell}\}_{j=1}^W$, where we are trying to solve
\[
\underbrace{
\begin{pmatrix}
s_1 & s_2 & \cdots & s_W \\
s_1 P_{1,1} & s_2 P_{2,1} & \cdots & s_W P_{W,1} \\
\vdots & \vdots & & \vdots \\
s_1 P_{1,1}^{W-1} & \cdots & \cdots & s_W P_{W,1}^{W-1}
\end{pmatrix}
}_{V}
\begin{pmatrix}
P_{1,\ell} \\ \vdots \\ P_{W,\ell}
\end{pmatrix}
= y.
\]
But
\[
V = \underbrace{
\begin{pmatrix}
1 & 1 & \cdots & 1 \\
P_{1,1} & P_{2,1} & \cdots & P_{W,1} \\
P_{1,1}^2 & P_{2,1}^2 & \cdots & P_{W,1}^2 \\
\vdots & \vdots & & \vdots \\
P_{1,1}^{W-1} & \cdots & \cdots & P_{W,1}^{W-1}
\end{pmatrix}
}_{U}
\cdot \operatorname{diag}(s_1, \ldots, s_W).
\]
Since $U$ is a Vandermonde matrix, we have
\begin{align*}
    \det(V) &= \det(U) \cdot \left( \prod_{j=1}^{W} s_j \right) \\ &= \prod_{1 \le a < b \le W} (P_{b,1} - P_{a,1}) \cdot \left( \prod_{j=1}^{W} s_j \right) \neq 0,
\end{align*}
where the inequality holds since $\{P_{j,1}\}_{j=1}^W$ are distinct.

Thus, under our assumption that $\{P_{j,1}\}_{j=1}^W$ are distinct, this linear system has a unique solution, and we can recover $(P_{j,\ell})_{j=1}^W$. Solving one system for each $\ell = 2, \ldots, d$ recovers all these unknowns.
\end{proof}

\begin{lemma}\label{lem_2}
If $R$ is a uniformly random unitary matrix, then the scalars $\langle r_1, x_1^* \rangle, \ldots, \langle r_1, x_W^* \rangle$ are distinct with probability 1.
\end{lemma}

\begin{proof}[Proof of Lemma~\ref{lem_2}]
Fix a pair $j\neq k$. Since $x_1^*,\ldots,x_W^*$ are distinct, we have $\Delta_{jk}:=x_j^*-x_k^*\neq 0$. If $\langle r_1, x_j^*\rangle = \langle r_1, x_k^*\rangle$, then we have the inclusion $r_1 \in \{r\in \mathbb{S}^{d-1} : \langle r, \Delta_{jk}\rangle = 0\}$. However, the set $\{r\in \mathbb{S}^{d-1} : \langle r, \Delta_{jk}\rangle = 0\}$ has measure zero under the uniform distribution on $\mathbb{S}^{d-1}$. Consequently, $\mathbb{P}(\langle r_1, x_j^*\rangle = \langle r_1, x_k^*\rangle) = 0$.
Since there are only $\binom{W}{2}$ distinct pairs, applying the union bound completes the proof.
\end{proof}

\begin{proof}[Proof of Theorem~\ref{thm:injectivity}] 
Combining the two lemmas yields that we can recover the scalars $\langle r_\ell, x_j^*\rangle$ for every $\ell \in [d]$ and $j \in [W]$ with probability $1$.
Since $r_1,\ldots,r_d$ form an orthonormal basis of $\mathbb{R}^d$,
we have
$x_j^* = \sum_{\ell=1}^d \langle r_\ell, x_j^*\rangle\, r_\ell$ for all $j \in [W]$.
We may thus invert this linear system to obtain each $x_j^*$ exactly. \end{proof}

\subsection{Proof of computational hardness}

As mentioned before, our proof proceeds via reduction from the $\mathsf{NP}$-hard \textsc{Partition} problem (also called the SubsetSum problem). The ideas are inspired by reductions to  other problems in statistical signal processing~\cite[e.g.][]{fickus2014phase,yi2014alternating,pananjady2017linear}.

\begin{proposition}\label{prop:partition}
Given a collection of $d$ nonzero rational\footnote{For convenience, we ignore complexity-theoretic issues with encoding real numbers; all our arguments can be made fully rigorous with suitable rounding.} numbers $(b_1, \ldots, b_d)$, the problem of deciding if there exists a subset $S \subseteq [d]$ such that $\sum_{k \in S} b_k = \sum_{k \in S^c} b_k$ is $\textsf{NP}$-hard.
\end{proposition}

\begin{proof}[Proof of Theorem~\ref{thm:NP-hard}]
Our reduction proceeds by constructing data  $(A, y)$ complying with the CT model~\eqref{eq:model-simple} from an instance of \textsc{Partition}. 
Consider an instance of the \textsc{Partition} problem $(b_1, \ldots, b_d)$, and let $n = 2d + 1$. We construct the $n \times d$ matrix
$A = \begin{bmatrix}
 I_d \\
 -I_d \\
 1^\top
\end{bmatrix}$ and the vector $y \in \real^n$ where
\begin{align} \label{eq:eq:y-LHS}
y_i = 
\begin{cases}
\exp(-b_i/2) \cdot \cosh(b_i/2) \quad &\text{ if } i = 1, \ldots, d \\
\exp(b_{i-d}/2) \cdot \cosh(b_{i-d}/2) \quad &\text{ if } i = d+1, \ldots, 2d \\
\exp(- (\sum_{i = 1}^d b_i)/2)  \quad &\text{ if } i = 2d + 1.
\end{cases}
\end{align}
To characterize the solution $(x^*_1, x^*_2)$ consistent with the data $(A, y)$, note that the RHS $R_i$ in model~\eqref{eq:model-simple} can be written as
\begin{align*}
R_i &= \exp( - \inprod{a_i}{x^*_1 + x^*_2} / 2) \cdot \\ &\left( \frac{1}{2} \exp( - \inprod{a_i}{x^*_1 - x^*_2}/2 ) + \frac{1}{2} \exp(\inprod{a_i}{x^*_1 - x^*_2}/2 ) \right) \\
&= \exp( - \inprod{a_i}{x^*_1 + x^*_2} / 2) \cdot \cosh ( \inprod{a_i}{x^*_1 - x^*_2}/2).
\end{align*}
For our particular choice of matrix $A$, we further have
\begin{align} \label{eq:RHS}
R_i =
\begin{cases}
  e^{-(x^*_1 + x^*_2)_i/2}\cosh((x^*_1 - x^*_2)_i/2), & i = 1,\ldots,d,\\
  e^{(x^*_1 + x^*_2)_{i-d}/2}\cosh((x^*_1 - x^*_2)_{i-d}/2), & i = d+1,\ldots,2d.
\end{cases}
\end{align}
For a solution to be consistent with the first $2d$ equations in the system, $y_i$ given by Eq.~\eqref{eq:eq:y-LHS} must be equated with the corresponding $R_i$ in Eq.~\eqref{eq:RHS} (for $i = 1, \ldots, d$). For each $i \in [d]$, considering the $i$-th and $(i + d)$-th equations together yields the following pair of simultaneous equations: 
\begin{align*}
&e^{-(x^*_1 + x^*_2)_i/2}\cosh\!\big((x^*_1 - x^*_2)_i/2\big) = e^{-b_i/2}\cosh(b_i/2),\\
&e^{(x^*_1 + x^*_2)_i/2}\cosh\!\big((x^*_1 - x^*_2)_i/2\big) = e^{b_i/2}\cosh(b_i/2).
\end{align*}

Since $\cosh(z)$ is always nonzero for real $z$, we have 
\[
\frac{\cosh(b_i/2)}{\cosh ( (x^*_1 - x^*_2)_i /2)} = \exp(b_i/2) \cdot \exp( - (x^*_1 + x^*_2)_i / 2)
\]
from the first equation. Substituting this into the second equation yields
$\exp((x^*_1 + x^*_2)_i) = \exp(b_i)$, which, owing to the invertibility of the map $z \mapsto e^z$, implies that 
$(x^*_1 + x^*_2)_i = b_i$. Backsubstituting and noting that $\cosh$ is an even function yields $(x^*_1 - x^*_2)^2_i = b^2_i$. Combining these two facts, for each $i \in [d]$ we must have 
\begin{align*}
\{ (x^*_1)_i = b_i \; \text{and} \; (x^*_2)_i = 0 \} \; \text{ or } \;  \{(x^*_1)_i = 0 \; \text{and} \; (x^*_2)_i = b_i \}.
\end{align*}
For any such choice, let $S = \mathsf{supp}(x^*_1)$, so that $S^c = \mathsf{supp}(x^*_2)$.
Using that for all $i \in [d]$ we have $(x^*_1 + x^*_2)_i = b_i$, let us now check the $(2d + 1)$-th equation. Since $a_{2d+1}$ is the all-ones vector, we have
\begin{align*}
R_{2d+1} &= \exp( - \inprod{1}{x^*_1 + x^*_2} / 2) \cdot \cosh ( \inprod{1}{x^*_1 - x^*_2}/2) \\ &= \exp \left( - (\sum_{i = 1}^d b_i)/2 \right) \cdot \cosh \left( (\sum_{i \in S} b_i - \sum_{i \in S^c} b_i ) / 2 \right).
\end{align*}
But $y_{2d+ 1} = \exp \left(- (\sum_{i = 1}^d b_i)/2 \right)$ by Eq.~\eqref{eq:eq:y-LHS}. If we have a consistent solution to the system, then $R_{2d + 1} = y_{2d + 1}$, or equivalently,
$
\cosh \left( (\sum_{i \in S} b_i - \sum_{i \in S^c} b_i ) / 2 \right) = 1,
$
thereby implying that we have found a set $S = \mathsf{supp}(x^*_1) \subseteq [d]$ with 
\[
\sum_{i \in S} b_i = \sum_{i \in S^c} b_i. 
\]
Any consistent solution to the CT system $(A, y)$ thus produces a ``YES" certificate for \textsc{Partition}, and this completes the reduction. Since our reduction was polynomial-time computable, Proposition~\ref{prop:partition} implies that deciding if a consistent solution exists in the model~\eqref{eq:model-simple} is $\mathsf{NP}$-hard.
\end{proof}

\section{Conclusion} \label{sec:conclusion}

In this work, we study the injectivity and computational complexity of the polychromatic CT model~\eqref{eq:model}. For injectivity, we leverage a construction inspired by Prony's method and show that by using $n = Wd + W - 1$ random measurements, the unknown signals can be uniquely recovered with probability $1$. For complexity, we prove that even with $W = 2$ wavelengths, the problem of signal recovery is $\mathsf{NP}$-hard by applying a reduction from the \textsc{Partition} problem.

Our results are stated for general measurement vectors $a_i$, but in computed tomography these vectors are far from arbitrary. In CT, the $a_i$ vectors are highly structured as they arise as discretized line integrals of the Radon transform \citep{natterer2001mathematics}.
Understanding analogous questions about injectivity and complexity in this structured setting remains open. Another natural question concerns the general multimaterial CT problem. While the results in our work address the case $M = W$, understanding injectivity thresholds in the general $M < W$ case remains an interesting open problem.

\subsection*{Acknowledgments}
This work was supported in part by the National Science Foundation through grant CCF-2107455, and a Google Research Scholar Award. We thank Mengqi Lou, Kabir Verchand, and Sara Fridovich-Keil for helpful discussions.

\bibliographystyle{abbrvnat}
\bibliography{refs}

\begin{thebibliography}{13}
\providecommand{\natexlab}[1]{#1}
\providecommand{\url}[1]{\texttt{#1}}
\expandafter\ifx\csname urlstyle\endcsname\relax
  \providecommand{\doi}[1]{doi: #1}\else
  \providecommand{\doi}{doi: \begingroup \urlstyle{rm}\Url}\fi

\bibitem[Barber and Sidky(2024)]{barber2024convergence}
R.~F. Barber and E.~Y. Sidky.
\newblock Convergence for nonconvex {ADMM}, with applications to {CT} imaging.
\newblock \emph{Journal of Machine Learning Research}, 25\penalty0 (38):\penalty0 1--46, 2024.

\bibitem[Charisopoulos and Willett(2025)]{charisopoulos2025nonlinear}
V.~Charisopoulos and R.~Willett.
\newblock Nonlinear tomographic reconstruction via nonsmooth optimization.
\newblock \emph{SIAM Journal on Mathematics of Data Science}, 7\penalty0 (2):\penalty0 699--722, 2025.

\bibitem[Fickus et~al.(2014)Fickus, Mixon, Nelson, and Wang]{fickus2014phase}
M.~Fickus, D.~G. Mixon, A.~A. Nelson, and Y.~Wang.
\newblock Phase retrieval from very few measurements.
\newblock \emph{Linear Algebra and its Applications}, 449:\penalty0 475--499, 2014.

\bibitem[Fridovich-Keil et~al.(2026)Fridovich-Keil, Valdivia, Wetzstein, Recht, and Soltanolkotabi]{fridovich2026gradient}
S.~Fridovich-Keil, F.~Valdivia, G.~Wetzstein, B.~Recht, and M.~Soltanolkotabi.
\newblock Gradient descent provably solves nonlinear tomographic reconstruction.
\newblock \emph{IEEE Transactions on Information Theory}, 2026.

\bibitem[Giesbrecht et~al.(2006)Giesbrecht, Labahn, and Lee]{giesbrecht2006symbolic}
M.~Giesbrecht, G.~Labahn, and W.-S. Lee.
\newblock Symbolic-numeric sparse interpolation of multivariate polynomials.
\newblock In \emph{Proceedings of the 2006 International Symposium on Symbolic and Algebraic Computation}, pages 116--123, 2006.

\bibitem[Kim et~al.(2026)Kim, Pananjady, Pourmorteza, and Fridovich-Keil]{kim2026perfusion}
N.~Kim, A.~Pananjady, A.~Pourmorteza, and S.~Fridovich-Keil.
\newblock Perfusion imaging and single material reconstruction in polychromatic photon counting {CT}.
\newblock \emph{arXiv preprint arXiv:2602.02713}, 2026.

\bibitem[Lou et~al.(2026)Lou, Verchand, Fridovich-Keil, and Pananjady]{lou2026accurate}
M.~Lou, K.~Verchand, S.~Fridovich-Keil, and A.~Pananjady.
\newblock Accurate, provable, and fast polychromatic tomographic reconstruction: A variational inequality approach.
\newblock \emph{SIAM Journal on Imaging Sciences}, 19\penalty0 (1):\penalty0 446--479, 2026.

\bibitem[Natterer(2001)]{natterer2001mathematics}
F.~Natterer.
\newblock \emph{The mathematics of computerized tomography}.
\newblock SIAM, 2001.

\bibitem[Pananjady et~al.(2017)Pananjady, Wainwright, and Courtade]{pananjady2017linear}
A.~Pananjady, M.~J. Wainwright, and T.~A. Courtade.
\newblock Linear regression with shuffled data: Statistical and computational limits of permutation recovery.
\newblock \emph{IEEE Transactions on Information Theory}, 64\penalty0 (5):\penalty0 3286--3300, 2017.

\bibitem[Papadimitriou and Steiglitz(1998)]{papadimitriou1998combinatorial}
C.~H. Papadimitriou and K.~Steiglitz.
\newblock \emph{Combinatorial optimization: Algorithms and complexity}.
\newblock Courier Corporation, 1998.

\bibitem[Schmidt et~al.(2022)Schmidt, Sammut, Barber, Pan, and Sidky]{https://doi.org/10.1002/mp.15621}
T.~G. Schmidt, B.~A. Sammut, R.~F. Barber, X.~Pan, and E.~Y. Sidky.
\newblock Addressing {CT} metal artifacts using photon-counting detectors and one-step spectral {CT} image reconstruction.
\newblock \emph{Medical Physics}, 49\penalty0 (5):\penalty0 3021--3040, 2022.
\newblock \doi{https://doi.org/10.1002/mp.15621}.
\newblock URL \url{https://aapm.onlinelibrary.wiley.com/doi/abs/10.1002/mp.15621}.

\bibitem[Tang(2023)]{tang2023spectral}
X.~Tang.
\newblock \emph{Spectral Multi-detector Computed Tomography (sMDCT): Data Acquisition, Image Formation, Quality Assessment and Contrast Enhancement}.
\newblock CRC Press, 2023.

\bibitem[Yi et~al.(2014)Yi, Caramanis, and Sanghavi]{yi2014alternating}
X.~Yi, C.~Caramanis, and S.~Sanghavi.
\newblock Alternating minimization for mixed linear regression.
\newblock In \emph{International Conference on Machine Learning}, pages 613--621. PMLR, 2014.

\end{thebibliography}

\end{document}